\RequirePackage[l2tabu,orthodox]{nag}
\documentclass[a4paper,11pt]{article} 	

\usepackage{color}

\usepackage[utf8]{inputenc}					
\usepackage[british]{babel}					
\usepackage[T1]{fontenc}					
\usepackage{hyperref}

\usepackage{amsmath,amssymb} 		
\usepackage{mathtools}						
\usepackage{setspace}
\usepackage[footnote,draft,danish,silent,nomargin]{fixme}
\usepackage{mathdots}

\usepackage[backend=bibtex, maxnames=99]{biblatex}

\usepackage{tikz}
\usetikzlibrary{patterns}

\usepackage[amsmath,thmmarks,hyperref]{ntheorem}
\theoremstyle{plain}
\theoremseparator{:}
\global
\global

\newtheorem{thm}{Theorem}[section]
\newtheorem*{nonumprop}{Proposition \ref{prop:sum_prod_cyclic}}
\newtheorem{defi}[thm]{Definition}
\newtheorem{cor}[thm]{Corollary}
\newtheorem{lem}[thm]{Lemma}
\newtheorem{prop}[thm]{Proposition}
\theorembodyfont{}

\theoremsymbol{\ensuremath{\blacktriangleleft}}
\newtheorem{ex}[thm]{Example}

\theoremstyle{nonumberbreak}
\theoremsymbol{}

\theoremsymbol{\tikz \fill[black] (0,0) rectangle (1.2ex,1.2ex);}
\newtheorem{proof}{Proof}

\usepackage{empheq}
\usepackage{pdflscape}

\newcommand{\vek}[1]{\boldsymbol{\mathbf{#1}}}

\bibliography{../Bibliografi.bib}

\hypersetup{pdftitle	={Squares of Matrix-product Codes},
			pdfauthor	={Ignacio Cascudo, Jaron Skovsted Gundersen, and Diego Ruano}
}

\title{Squares of Matrix-product Codes\thanks{This work is supported by the Danish Council for Independent Research: grant DFF-4002-00367, the Spanish Ministry of Economy/FEDER: grant RYC-2016-20208 (AEI/FSE/UE), the Spanish Ministry of
Science/FEDER: grant PGC2018-096446-B-C21, and Junta de CyL (Spain): grant VA166G18.}
}

\author{Ignacio Cascudo\thanks{IMDEA Software Institute, Madrid, Spain. Email: ignacio.cascudo@imdea.org}, Jaron Skovsted Gundersen\thanks{Department of Mathematical Sciences, Aalborg University, Denmark. Email: jaron@math.aau.dk}, and Diego Ruano\thanks{IMUVA-Mathematics Research Institute, Universidad de Valladolid, Spain. Email: diego.ruano@uva.es}}

\begin{document}

\maketitle

\begin{abstract}
\noindent The component-wise or Schur product $C*C'$ of two linear error-correcting codes $C$ and $C'$ over certain finite field is the linear code spanned by all component-wise products of a codeword in $C$ with a codeword in $C'$. When $C=C'$, we call the product the square of $C$ and denote it $C^{*2}$. Motivated by several applications of squares of linear codes in the area of cryptography, in this paper we study squares of so-called matrix-product codes, a general construction that allows to obtain new longer codes from several ``constituent'' codes. We show that in many cases we can relate the square of a matrix-product code to the squares and products of their constituent codes, which allow us to give bounds or even determine its minimum distance. We consider the well-known $(u,u+v)$-construction, or Plotkin sum (which is a special case of a matrix-product code) and determine which parameters we can obtain when the constituent codes are certain cyclic codes. In addition, we use the same techniques to study the squares of other matrix-product codes, for example when the defining matrix is Vandermonde (where the minimum distance is in a certain sense maximal with respect to matrix-product codes).
\end{abstract}

\section{Introduction}
Component-wise or Schur products of linear error-correcting codes have been studied for different purposes during the last decades, from efficient decoding to applications in several different areas within cryptography. Given two linear (over some finite field $\mathbb{F}$) codes $C_1,C_2$ of the same length we define the component-wise product of the codes $C_1*C_2$ to be the span over $\mathbb{F}$ of all component-wise products $\vek{c}_1*\vek{c}_2$, where $\vek{c}_i\in C_i$.

Many of these applications involve the study of the main parameters of interest in the theory of error-correcting codes, namely dimension and minimum distance, where in this case we want to analyze the behaviour of some of these parameters both on the factor codes and their product simultaneously. Recall that the dimension $\dim C$ of a linear code $C$ is its dimension as a vector space over the finite field, and the minimum distance $d(C)$ of $C$ is the smallest Hamming distance between two distinct codewords of the code. Additionally, some applications involve the analysis of the minimum distance of the duals of the factor codes (recall that the dual $C^{\perp}$ of $C$ is the linear code containing all vectors that are orthogonal to all codewords of $C$ under the standard inner product over $\mathbb{F}$). 
 
One of the first applications where component-wise products of codes became relevant concerned error decoding via the notion of error-locating pairs in the works of Pellikaan and of Duursma and Kötter~\cite{DK94,P92}. An error-locating pair for a code $C$ is a pair $C_1,C_2$ where $C_1*C_2\subseteq C^\perp$, and the number of errors the pair is able to correct depends on the dimensions and minimum distances of the codes and their duals. More precisely, it is required that $\dim(C_1)>t$ and $d(C_2^\perp)>t$ if we should be able to locate $t$ errors.

Later on, the use of component-wise products found several applications in the area of cryptography. For example, some attacks to variants of the McEliece cryptosystem (which relies on the assumption that it is hard to decode a general linear code) use the fact that the dimension of the product $C*C$ tends to be much larger when $C$ is a random code than when $C$ has certain algebraic structure, which can be used to identify algebraic patterns in certain subcodes of the code defining the cryptosystem, see for instance \cite{CMP17, COT17, MS07, PM17, W10}. 

A different cryptographic problem where products of codes are useful is private information retrieval, where a client can retrieve data from a set of servers allocating a coded database in such a way that the servers do not learn what data the client has accessed. In \cite{FGHK17} a private information retrieval protocol based on error-correcting codes was shown, where it is desirable to use two linear codes $C_1$ and $C_2$ such that $\dim(C_1)$, $d(C_2^\perp)$, and $d(C_1*C_2)$ are simultaneously high. 

In this work, however, we are more interested in the application of products of codes to the area of secure multiparty computation. The goal of secure multiparty computation is to design protocols which can be used in the situation where a number of parties, each holding some private input, want to jointly compute the output of a function on those inputs, without at any point needing that any party reveals his/her input to anybody. A central component in secure computation protocols is secure multiplication, which different protocols realize in different ways. Several of these protocols require to use an error-correcting code $C$ whose square $C^{*2}=C*C$ has large minimum distance while there are additional conditions on $C$ which vary across the different protocols.

For example a well known class of secure computation protocols \cite{BGW88, CCD88, CDM00} relies on the concept of strongly multiplicative secret sharing scheme formalized in \cite{CDM00}. Such secret sharing schemes can be constructed from linear codes $C$ where the amount of colluding cheating parties that the protocol can tolerate 
is $\min (d(C^\perp), d(C^{*2}))+2$, where $C^{\perp}$ is the dual code to $C$. These two minimum distances are therefore desired to be simultaneously high. For more information about secret sharing and multiparty computation, see for instance \cite{CDN15}.

Other more recent protocols have the less stringent requirement that $d(C^{*2})$ and $\dim C$ are simultaneously large. This is the case of the MiniMac \cite{DZ13} protocol, a secure computation protocol to evaluate boolean circuits, and its successor Tinytables \cite{DNNR16}. In those protocols, the cheating parties have certain probability of being able to disrupt the computation, but this probability is bounded by $2^{-d(C^{*2})}$, meaning that a high distance on the square will give a higher security. On the other hand, a large relative dimension, or rate, of $C$ will reduce the communication cost, so it is desirable to optimize both parameters.
A very similar phenomenon occurs in recent work about commitment schemes, which are a building block of many multiparty computation protocols; in fact, when these schemes have a number of additional homomorphic properties and in addition can be composed securely, we can base the entire secure computation protocol on them \cite{FPY18}. Efficient commitment schemes with such properties were constructed in \cite{CDD18} based on binary linear codes, where multiplicative homomorphic properties require again to have a relatively large $d(C^{*2})$ (see \cite[section 4]{CDD18}) and the rate of the code is also desired to be large to reduce the communication overhead.

These applications show the importance of finding linear codes where the minimum distance of the square $d(C^{*2})$ is large relative to the length of the codes and where some other parameter (in some cases $\dim(C)$, in others $d(C^\perp)$) is also relatively large. Moreover, it is especially interesting for the applications that the codes are binary, or at least be defined over small fields.

Powers of codes, and more generally products, have been studied in several works such as \cite{C17, CCMZ15, MZ15, R13b, R13a, R15} from different perspectives. In \cite{R13b} an analog of the Singleton bound for $\dim(C)$ and $d(C^{*2})$ was established, and in \cite{MZ15} it is shown that Reed Solomon codes are essentially the only codes which attain this bound unless some of the parameters are very restricted. However, Reed Solomon codes come with the drawback that the field size must be larger than or equal to the length of the codes. Therefore, finding asymptotically good codes over a fixed small field has also been studied, where in this case asymptotically good means that both $\dim(C)$ and $d(C^{*2})$ grows linearly with the length of the code $C$. In \cite{R13a} the existence of such a family over the binary field was shown, based on recent results on algebraic function fields. However, it seems like most families of codes do not have this property: in fact, despite the well known fact that random linear codes will, with high probability, be over the Gilbert Varshamov bound, and hence are asymptotically good in the classical sense, this is not the case when we impose the additional restriction that $d(C^{*2})$ is linear in the length, as it is shown in \cite{CCMZ15}. The main result in \cite{CCMZ15} implies that for a family of random linear codes either the code or the square will be asymptotically bad.

The asymptotical construction from \cite{R13a}, despite being very interesting from the theoretical point of view, has the drawbacks that the asymptotics kick in relatively late and moreover, the construction relies on algebraic geometry, which makes it computationally expensive to construct such codes. Motivated by the aforementioned applications to cryptography, \cite{C17} focuses on codes with fixed lengths (but still considerably larger than the size of the field), and constructs cyclic codes with relatively large dimension and minimum distance of their squares. In particular, the parameters of some of these codes are explicitly computed in the binary case. 

This provides a limited constellation of parameters that we know that are achievable for the tuple consisting of length of $C$, $\dim(C)$ and $d(C^{*2})$. It is then interesting to study what other parameters can be attained, and a natural way to do so is to study how the square operation behaves under known procedures in coding theory that allow to construct new codes from existing codes.

One such construction is matrix-product codes, where several codes can be combined into a new longer code. Matrix-product codes, formalized in \cite{BN01}, is a generalization of some previously known code constructions, such as for example the $(u,u+v)$-construction, also known as Plotkin sum.  Matrix-product codes have been studied in several works, including \cite{BN01, HLR09, HR10, OS02}.

\subsection{Results and outline}

In this work, we study squares of matrix-product codes. We show that in several cases, the square of a matrix-product code can be also written as a matrix-product code. This allows us to determine new achievable parameters for the squares of codes.

More concretely, we start by introducing matrix-product codes and products of codes in Section \ref{sec:prelim}. Afterwards, we determine the product of two codes when both codes are constructed using the $(u,u+v)$-construction in Section \ref{sec:u,u+v}. In Section \ref{sec:cyclic}, we restrict ourselves to squares of codes and exemplify what parameters we can achieve using cyclic codes in the $(u,u+v)$-construction in order to compare the parameters with the codes from \cite{C17}. 

At last, in Section \ref{sec:other_cons}, we consider other constructions of matrix-product codes. In particular, we consider the case where the defining matrix is Vandermonde, which is especially relevant because such matrix-product codes achieve the best possible minimum distance that one can hope for with this matrix-product strategy. We show that the squares of these codes are again matrix-product codes, and if the constituent codes of the original matrix-product code are denoted $C_i$, then the ones for the square are all of the form $\sum_{i+j=l}C_i*C_j$ for some $l$. This is especially helpful for determining the parameters if the $C_i$'s are for example algebraic geometric codes. We remark that this property also holds for the other constructions we study in this paper, but only when the $C_i$'s are nested. Finally, we also study the squares of a matrix-product construction from \cite{BN01} where we can apply the same proof techniques as we have in the other constructions.

\section{Preliminaries}\label{sec:prelim}
Let $\mathbb{F}_q$ be the finite field with $q$ elements. A linear code $C$ is a subspace of $\mathbb{F}_q^n$. When $C$ has dimension $k$, we will call it an $[n,k]_q$ code. A generator matrix for a code $C$ is a $k\times n$ matrix consisting of $k$ basis vectors for $C$ as the $k$ rows. The Hamming weight of $\vek{x}\in \mathbb{F}_q^n$, denoted $w(\vek{x})$, is the number of nonzero entries in $\vek{x}$ and the Hamming distance between $\vek{x},\vek{y}\in \mathbb{F}_q^n$ is given by $d(\vek{x},\vek{y})=w(\vek{x}-\vek{y})$. By the linearity of $C$ the minimum Hamming distance taken over all pairs of distinct elements in $C$ is the same as the minimum Hamming weight taking over all non-zero elements in $C$, and therefore we define the minimum distance of $C$ to be	$d(C)=\min_{\vek{x}\in C\setminus\{0\}}\{w(\vek{x})\}$. If it is known that $d(C)=d$ (respectively if we know that $d(C)\geq d$) then we call $C$ an $[n,k,d]_q$ code (resp. $[n,k,\geq d]_q$). We denote by $C^{\perp}$ the dual code to $C$, i.e., the vector space given by all elements $\vek{y}\in \mathbb{F}_q^n$ such that for every $\vek{x}\in C$, $\vek{x}$ and $\vek{y}$ are orthogonal with respect to the standard inner product in $\mathbb{F}_q^n$. If 
$C$ is an $[n,k]_q$ code then $C^{\perp}$ is an $[n,n-k]_q$ code.

We recall the definition and basic properties of matrix-product codes (following \cite{BN01}) and squares of codes. 

	\begin{defi}[Matrix-product code]\label{def:mp_codes}
		Let $C_1, \ldots, C_s \subseteq \mathbb{F}_q^n$ be linear codes and let $A\in \mathbb{F}_q^{s\times l}$ be a matrix with rank $s$ (implying $s\leq l$). Then we define the matrix-product code $C = [C_1,\ldots,C_s]A$, as the set of all matrix products $[\vek{c}_1,\ldots,\vek{c}_s]A$, where $\vek{c}_i=(c_{1i},\ldots,c_{ni})^T\in C_i$. 
	\end{defi}
	We call $A=\left[a_{ij}\right]_{i=1,\ldots,s, j=1,\ldots,l}$ the defining matrix and the $C_i$'s the constituent codes. We can consider a codeword $\vek{c}$, in a matrix-product code, as a matrix of the form
	\begin{align}\label{eq:matrix_rep}
		\vek{c}=\begin{bmatrix}
			c_{11}a_{11}+c_{12}a_{21}+\cdots+c_{1s}a_{s1} & \cdots & c_{11}a_{1l}+c_{12}a_{2l}+\cdots+c_{1s}a_{sl} \\
			\vdots & \ddots & \vdots\\
			c_{n1}a_{11}+c_{n2}a_{21}+\cdots+c_{ns}a_{s1} & \cdots & c_{n1}a_{1l}+c_{n2}a_{2l}+\cdots+c_{ns}a_{sl}
		\end{bmatrix},
	\end{align}
using the same notation for the $\vek{c}_i$'s as in the definition. Reading the entries in this matrix in a column-major order, we can also consider $\vek{c}$ as a vector of the form
	\begin{align}\label{eq:vector_rep}
		\vek{c}=\left(\sum_{i=1}^s a_{i1}\vek{c}_i,\ldots,\sum_{i=1}^s a_{il}\vek{c}_i\right)\in \mathbb{F}_q^{nl}.
	\end{align}
We sum up some known facts about matrix-product codes in the following proposition. 
	\begin{prop}\label{prop:matrix_prod_codes}
		Let $C_1,\ldots, C_s \subseteq \mathbb{F}_q^n$ be linear $[n,k_1],\ldots,[n,k_s]$ codes with generator matrices $G_1,\ldots,G_s$, respectively. Furthermore, let $A\in \mathbb{F}_q^{s\times l}$ be a matrix with rank $s$ and let $C=[C_1,\ldots,C_s]A$. Then $C$ is an $[nl,k_1+\cdots+k_s]$ linear code and a generator matrix of $C$ is given by
		\begin{align*}
			G= \begin{bmatrix}
				a_{11}G_1 & \cdots & a_{1l}G_1\\
				\vdots & \ddots & \vdots\\
				a_{s1}G_s & \cdots & a_{sl}G_s
			   \end{bmatrix}.
		\end{align*}
	\end{prop}
We now turn our attention to the minimum distance of $C$. Denote by $A_i$ the matrix consisting of the first $i$ rows of $A$ and let $C_{A_i}$ be the linear code spanned by the rows in $A_i$. From \cite{OS02}, we have the following result on the minimum distance.
	\begin{prop}\label{prop:min_dist_bound}
		We are making the same assumptions as in Proposition \ref{prop:matrix_prod_codes}, and write $D_i=d(C_{A_i})$ and $d_i=d(C_i)$. Then the minimum distance of the matrix-product code $C$ satisfies
		\begin{align}\label{eq:min_dist_bound}
			d(C)\geq\min\{D_1d_1,D_2d_2,\ldots,D_sd_s\}.
		\end{align}
	\end{prop}
	The following corollary is from \cite{HLR09}.
	\begin{cor}\label{cor:nested}
		If we additionally assume that $C_1\supseteq \cdots \supseteq C_s$, equality occurs in the bound in \eqref{eq:min_dist_bound}.
	\end{cor}
	The dual of a matrix-product code is also a matrix-product code, if we make some assumptions on the matrix $A$, as it was noted in \cite{BN01}.
	\begin{prop}\label{prop:dual}
	Let $C=[C_1,C_2,\ldots,C_s]A$ be a matrix product code. If $A$ is an invertible square matrix then 
	\begin{align*}
		C^\perp=[C_1^\perp,C_2^\perp,\ldots,C_s^\perp ] (A^{-1})^T
	\end{align*}
	Additionally, if $J$ is the $s\times s$ matrix given by
	\begin{align*}
		J=\begin{bmatrix}
		0 & \cdots & 0 & 1\\
		0 & \cdots & 1 & 0\\
		\vdots & \iddots & \vdots & \vdots\\
		1 & \cdots & 0 & 0\\
		\end{bmatrix}
	\end{align*}
	the dual can be described as
	\begin{align*}
		C^\perp=[C_s^\perp,C_{s-1}^\perp,\ldots,C_1^\perp ] (J(A^{-1})^T).
	\end{align*}
	\end{prop}
	Notice that with regard to Proposition \ref{prop:min_dist_bound} the last expression is often more useful since $d(C_i^\perp)$ will often decrease when $i$ increases. 
	
	Now, we turn our attention to products and squares of codes. We denote by $*$ the component-wise product of two vectors. That is, if $\vek{x}=(x_1,\ldots,x_n)$ and $\vek{y}=(y_1,\ldots,y_n)$, then $\vek{x}*\vek{y}=(x_1y_1,\ldots,x_ny_n)$. With this definition in mind, we define the product of two linear codes.
	
	\begin{defi}[Component-wise (Schur) products and squares of codes]
		Given two linear codes $C,C'\subseteq \mathbb{F}_q^n$ we define their component-wise product, denoted by $C*C'$, as
		\begin{align*}
			C*C'=\langle \{\vek{c}*\vek{c}'\mid \vek{c}\in C, \vek{c}'\in C'\}\rangle.
		\end{align*}
		
The square of a code $C$ is $C^{*2}=C*C$.		
	\end{defi}
	First note that the length of the product is the same as the length of the original codes. Regarding the other parameters (dimension and minimum distance) we enumerate some known results only in the case of the squares $C^{*2}$ since this will be our primary focus.
		
	If 
	\begin{align*}
		G=\begin{bmatrix}
		\vek{g}_1\\
		\vek{g}_2\\
		\vdots\\
		\vek{g}_k
		\end{bmatrix}
	\end{align*}
	is a generator matrix for $C$, then $\{\vek{g}_i*\vek{g}_j\mid 1\leq i\leq j\leq n\}$ is a generating set for $C^{*2}$. However, it might not be a basis since some of the vectors might be linearly dependent. If additionally, a submatrix consisting of $k$ columns of $G$ is the identity, the set $\{ \vek{g}_i*\vek{g}_i\}$ consists of $k$ linearly independent vectors. Since there is always a generator matrix satisfying this, this implies $k\leq \dim(C^{*2})\leq \frac{k(k+1)}{2}$, where $k=\dim(C)$, and $d(C^{*2})\leq d(C)$. 
	
In most cases, however, $d(C^{*2})$ is much smaller than $d(C)$. For example, the Singleton bound for squares \cite{R13b} states that $d(C^{*2})\leq \max\{1,n-2k+2\}$ (which is much restrictive than the Singleton bound for $C$, which states that $d(C)\leq n-k+1$). Additionally, the codes for which $d(C^{*2})=n-2k+2$ have been characterized in \cite{MZ15}, where it was shown that essentially only Reed-Solomon codes, certain direct sums of self-dual codes, and some degenerate codes have this property. Furthermore, it is shown in \cite{CCMZ15} that taking a random code with dimension $k$ the dimension of $C^{*2}$ will with high probability be $\min\{n, \frac{k(k+1)}{2}\}$. Therefore, often $\dim(C^{*2})\gg \dim(C)$ and hence typically $d(C^{*2})\ll d(C)$. 

\section{The $(u,u+v)$-Construction}\label{sec:u,u+v}
In this section, we will consider one of the most well-known matrix-product codes, namely the $(u,u+v)$-construction. We obtain this construction when we let 
	\begin{align}\label{eq:u,u+v_matrix}
		A=\begin{bmatrix}
		1 & 1 \\ 0 & 1
		\end{bmatrix}\in \mathbb{F}_q^{2\times 2}
	\end{align}
	be the defining matrix. Note that if $C=\left[C_1,C_2\right]A$ then $d(C)\geq \min\{2d(C_1),d(C_2)\}$ and $d(C^\perp)\geq \min \{2d(C_2^\perp),d(C_1^\perp)\}$. This can easily be deduced from Propositions \ref{prop:min_dist_bound} and \ref{prop:dual} by constructing $J(A^{-1})^T$.
	
	In the following theorem, we will determine the product of two codes $C$ and $C'$ when both codes come from the $(u,u+v)$-construction. We will use the notation $C+C'$ to denote the smallest linear code containing both $C$ and $C'$. 
	\begin{thm}\label{thm_u_u+v_con}
		Let $C_1,C_2,C_1',C_2' \subseteq \mathbb{F}_q^n$ be linear codes. Furthermore, let $A$ be as in \eqref{eq:u,u+v_matrix} and denote by $C=[C_1,C_2]A$ and by $C'=[C_1',C_2']A$. Then 
		\begin{align*}
		C*C'=[C_1*C_1',C_1*C_2'+C_2*C_1'+C_2*C_2']A.
		\end{align*}
	\end{thm}
	\begin{proof}
	Let $G_1,G_2,G_1',G_2'$ be generator matrices for $C_1,C_2,C_1',C_2'$ respectively. By Proposition \ref{prop:matrix_prod_codes}, we have that
	\begin{align*}
		G=\begin{bmatrix}
		G_1 & G_1\\
		0 & G_2
		\end{bmatrix}, \quad G'=\begin{bmatrix}
		G_1' & G_1'\\
		0 & G_2'
		\end{bmatrix}
	\end{align*}
	are generator matrices for $C$ and $C'$ respectively. A generator matrix for $C*C'$ can be obtained by making the componentwise products of all the rows in $G$ with all the rows in $G'$ and afterwards removing all linearly dependent rows. We denote by $G*G'$ the matrix consisting of all componentwise products of rows in $G$ with rows in $G'$. Then
	\begin{align*}
		G*G'=\begin{bmatrix}
		G_1*G'_1 & G_1*G_1'\\ 0 & G_1*G_2'\\ 0 & G_2*G_1'\\ 0 & G_2*G_2'
		\end{bmatrix}.
	\end{align*}
The set of rows in $G_i*G_j'$ is a generating set for $C_i*C'_j$. Hence, by removing linearly dependent rows we obtain a generator matrix of the form
	\begin{align*}
		\tilde{G}=\begin{bmatrix}
		\tilde{G}_1 & \tilde{G}_1\\ 0 & \tilde{G}_2
		\end{bmatrix},
	\end{align*}
where $\tilde{G}_1$ is a generator matrix for $C_1*C_1'$, and $\tilde{G}_2$ for $C_1*C_2'+C_2*C_1'+C_2*C_2'$. By using Proposition \ref{prop:matrix_prod_codes} once again, we see that $\tilde{G}$ is a generator matrix for the code $[C_1*C_1',C_1*C_2'+C_2*C_1'+C_2*C_2']A$	proving the theorem. 
	\end{proof}
	The following corollary consider the square of a code from the $(u,u+v)$-construction, and in the remaining of the paper the focus will be on squares. 
	\begin{cor}\label{cor_u_u+v_con}
		Let $C_1,C_2 \subseteq \mathbb{F}_q^n$ be linear codes.	Furthermore, let $A$ be as in \eqref{eq:u,u+v_matrix} and denote by $C=[C_1,C_2]A$. Then 
		\begin{align*}
		C^{*2}=[C_1^{*2},(C_1+C_2)*C_2]A,
		\end{align*}
		and we have that
		\begin{align}\label{eq:min_dist:u_u_plus_v_not_nested}
			d(C^{*2})\geq \min\{2d(C_1^{*2}),d((C_1+C_2)*C_2)\}.
		\end{align}
		Additionally, if $C_2\subseteq C_1$ we obtain
		\begin{align*}
			C^{*2}=[C_1^{*2},C_1*C_2]A,
		\end{align*}
		and we have that
		\begin{align}\label{eq:min_dist_u_u_plus_n_nested}
			d(C^{*2})=\min\{2d(C_1^{*2}),d(C_1*C_2)\}.
		\end{align}
	\end{cor}	
	\begin{proof}
	The results follows by setting $C'_i=C_i$ in Theorem \ref{thm_u_u+v_con} implying $C'=C$, and we obtain that
	\begin{align*}
		C^{*2}=[C_1*C_1,C_1*C_2+C_2*C_1+C_2*C_2]A=[C_1^{*2},(C_1+C_2)*C_2]A.
	\end{align*}
	If $C_2\subseteq C_1$ we have $C_1+C_2=C_1$. The bound in \eqref{eq:min_dist:u_u_plus_v_not_nested} follows directly from Proposition \ref{prop:min_dist_bound}, and \eqref{eq:min_dist_u_u_plus_n_nested} follows by Corollary \ref{cor:nested}.
	\end{proof}
	
	\section{Constructions from Binary Cyclic Codes}\label{sec:cyclic}
In this section, we exemplify what parameters we can achieve for $C$ and $C^{*2}$ when we use the $(u,u+v)$-construction together with cyclic codes as constituent codes. We start by presenting some basics of cyclic codes. 

Cyclic codes are linear codes which are invariant under cyclic shifts. That is, if $\vek{c}=(c_0,c_1,\ldots,c_{n-2},c_{n-1})$ is a codeword then $(c_{n-1},c_0,c_1,\ldots,c_{n-2})$ is as well. We will assume that $\gcd(n,q)=1$. A cyclic code of length $n$ over $\mathbb{F}_q$ is isomorphic to an ideal in $R=\mathbb{F}_q[x]/\langle x^n-1\rangle$ generated by a polynomial $g$, where $g|x^n-1$. The isomorphism is given by
	\begin{align*}
		c_0+c_1x+\ldots+c_{n-1}x^{n-1}\mapsto (c_0,c_1,\ldots,c_{n-1}),
	\end{align*}
	and we notice that a cyclic shift is represented by multiplying by $x$. The cyclic code generated by $g$ has dimension $n-\deg g$. To bound the minimum distance of the code, we introduce the $q$-cyclotomic cosets modulo $n$.
	\begin{defi}[$q$-cyclotomic coset modulo $n$]
		Let $a\in \mathbb{Z}/n\mathbb{Z}$. Then the $q$-cyclotomic coset modulo $n$ of $a$ is given by
		\begin{align*}
			[a]=\{aq^j \bmod n \mid j\geq 0\}.
		\end{align*}
	\end{defi}
	Now let $\beta^n=1$ and $\beta^k\neq 1$ for $1\leq k\leq n$, meaning that $\beta$ is a primitive $n$-th root of unity in an algebraic closure of $\mathbb{F}_q$. Since $g|x^n-1$ every root of $g$ must be of the form $\beta^j$ for some $j\in\{0,1,\ldots,n-1\}=\mathbb{Z}/n\mathbb{Z}$. This leads to the following definition which turns out to be useful in describing the parameters of a cyclic code.
	\begin{defi}[Defining and generating set]
		Denote by $J= \{j \in \mathbb{Z}/n\mathbb{Z} \mid g(\beta^j) = 0\}$ and $I=\{j \in \mathbb{Z}/n\mathbb{Z} \mid g(\beta^j) \neq 0\}$. Then we call $J$ the defining set and $I$ the generating set of the cyclic code $C$ generated by $g$.
	\end{defi}
	We remark that $g=\prod_{j\in J}(x-\beta^j)=(x^n-1)/\prod_{i\in I}(x-\beta^i)$, implying that $|I|$ is the dimension of the cyclic code generated by $g$. We note that $I$ and $J$ must be a union of $q$-cyclotomic cosets modulo $n$. Now we define the amplitude of $I$ as
	\begin{align*}
		\mathrm{Amp}(I)=\min\{i\in \mathbb{Z}\mid \exists c\in \mathbb{Z}/n\mathbb{Z} \text{ such that } I\subseteq\{c,c+1,\ldots,c+i-1\}\}.
	\end{align*}
	As a consequence of the BCH-bound, see for example \cite{C17}, we have that the minimum distance of the code generated by $g$ is greater than or equal to $n-\mathrm{Amp}(I)+1$. 

Hence, we see that both the dimension and minimum distance depend on $I$, and since $C$ is uniquely determined by $I$, we will use the notation $C(I)$ to describe the cyclic code generated by $g=(x^n-1)/(\prod_{i\in I}(x-\beta^i))$. To summarize, we have that $C(I)$ is a cyclic linear code with parameters
\begin{align}\label{eq:param_cyclic}
	[n, |I|, \geq n- \mathrm{Amp}(I)+1]_q.
\end{align}
The dual of a cyclic code is also a cyclic code. In fact, $C(I)^{\perp}=C(-J)$, where $-J=\{-j \bmod n\mid j\in J\}$. Clearly $\mathrm{Amp}(-J)=\mathrm{Amp}(J)$ which shows that $C(I)^{\perp}$ is a cyclic linear code with parameters
\begin{align}\label{eq:param_cyclic_dual}
	[n, |J|, \geq n- \mathrm{Amp}(J)+1]_q.
\end{align}
Furthermore, we remark that $\mathrm{Amp}(J)=n-\max\{s\mid \{i,i+1,\ldots,i+s-1\}\subseteq I \text{ for some } i\}$, i.e $n$ minus the size of the largest set of consecutive elements in $I$. We conclude that the minimum distance of $C(I)^{\perp}$ is strictly larger than the size of any set of consecutive elements in $I$.

We consider cyclic codes for the $(u,u+v)$-construction, and therefore we will need the following proposition.
	\begin{prop}\label{prop:sum_prod_cyclic}
		Let $I_1$ and $I_2$ be unions of $q$-cyclotomic cosets, and let $C(I_1)$ and $C(I_2)$ be the corresponding cyclic codes. Then 
		\begin{align*}
			C(I_1)+C(I_2)&=C(I_1\cup I_2)\\
			C(I_1)*C(I_2)&= C(I_1+I_2),
		\end{align*}
		where $I_1+I_2=\{a+b \bmod n\mid a\in I_1,b\in I_2\}$.
	\end{prop}
We obtain this result by describing the cyclic codes as a subfield subcode of an evaluation code and generalizing Theorem 3.3 in \cite{C17}. The proof of this proposition is very similar to the one in \cite{C17} and can be found in Appendix \ref{sec:app}.	The proposition implies the following corollary.
	\begin{cor}\label{cor:param_cyclic}
		Let $I_1$ and $I_2$ be unions of $q$-cyclotomic cosets, and let $C(I_1)$ and $C(I_2)$ be the corresponding cyclic codes. Then $C(I_1)*C(I_2)$ is an 
		\begin{align*}
			[n, |I_1+I_2|, \geq n- \mathrm{Amp}(I_1+I_2)+1]_q
		\end{align*}
		cyclic code.
	\end{cor}
Now, let $C(I_1)$ and $C(I_2)$ be two cyclic codes over $\mathbb{F}_q$ of length $n$, and let 
	\begin{align}\label{eq:cyclic_u_u+v}
		C=[C(I_1),C(I_2)]\begin{bmatrix}
		1 & 1 \\ 0 & 1
		\end{bmatrix}.
	\end{align}
	Then $C$ is a
	\begin{align*}
		[2n, |I_1|+|I_2|,\geq \min\{2(n-\mathrm{Amp}(I_1)+1),n-\mathrm{Amp}(I_2)+1\}]_q
	\end{align*}
	linear code. This is in fact a quasi-cyclic code of index $2$, see for instance \cite{HLR09, LF01}.
By combining Corollary \ref{cor_u_u+v_con} with Proposition \ref{prop:sum_prod_cyclic}, we obtain that
	\begin{align}\label{eq:inclusion}
		C^{*2}= [C(I_1+I_1), C(I_2+(I_1\cup I_2))]\begin{bmatrix}
		1 & 1 \\ 0 & 1
		\end{bmatrix}.
	\end{align}
And from Propositions \ref{prop:matrix_prod_codes} and \ref{prop:min_dist_bound}, and Corollary \ref{cor:param_cyclic}, we obtain that
	\begin{align*}
		\dim(C^{*2})= |I_1+I_1|+|I_2+(I_1\cup I_2)|,
	\end{align*}
and
	\begin{align}\label{eq:dist_square_cyclic}
		d(C^{*2})\geq \min\{2(n- \mathrm{Amp}(I_1+I_1)+1),n- \mathrm{Amp}(I_2+(I_1\cup I_2))+1) \}.
	\end{align}
	Therefore, it is of interest to find $I_1$ and $I_2$ such that the cardinalities of these sets are relatively large, implying a large dimension of $C$, while at the same time $\mathrm{Amp}(I_1+I_1)$ and $\mathrm{Amp}(I_2+(I_1\cup I_2))$ are relatively small, implying a large minimum distance on the square. 

To exemplify what parameters we can obtain we will use some specific cyclic codes from \cite{C17} based on the notion of $s$-restricted weights of cyclotomic cosets introduced in the same article. Let $n=q^r-1$ for some $r$ and for a number $t\in \{0,1,\ldots,n-1\}$ let $(t_{r-1},t_{r-2},\ldots,t_0)$ be its $q$-ary representation, i.e. $t=\sum_{i=0}^{r-1}t_iq^i$, where $t_i\in\{0,1,\ldots,q-1\}$. Then for an $s\leq r$ the $s$-restricted weight is defined as
\begin{align*}
	w_q^{(s)}(t)=\max_{i\in \{0,1,\ldots,r-1\}} \sum_{j=0}^{s-1} t_{i+j}.
\end{align*}
We will not go into details about these $s$-restricted weights but we refer the reader to \cite{C17} for more information. However, we remark that \cite{C17} proves that this weight notion satisfies $w_q^{(s)}(v)\leq w_q^{(s)}(t)+w_q^{(s)}(u)$ if $v=t+u$, and that $w_q^{(s)}(t)=w_q^{(s)}(u)$ whenever $t$ and $u$ are in the same cyclotomic coset. The latter implies that we can talk about the $s$-restricted weight of a cyclotomic coset.

Let $W_{r,s,m}$ denote the union of all cyclotomic cosets modulo $q^r-1$ with $s$-restricted weights lower than or equal to $m$. I.e.
\begin{align*}
	W_{r,s,m}=\{t\in \{0,1,\ldots,q^r-2\}\mid w_q^{(s)}(t)\leq m\}.
\end{align*}
We remark that the minimum distance of $C(W_{r,s,m})$ can be deduced using that an upper bound for $\mathrm{Amp}(W_{r,s,m})$ is $\max W_{r,s,m}+1$ and the maximum element of $W_{r,s,m}$ can be easily deduced \cite{C17}. Furthermore, the dimension of $C(W_{r,s,m})$ is simply the cardinality of $W_{r,s,m}$, which either can be counted for the specific choices of $r$, $s$, and $m$, or can be expressed as a recurrent sequence in $r$ (for a fixed selection of adequately small $s$ and $m$) using an argument involving counting closed walks of length $r$ in certain graph, see \cite{C17}. Finally, it is not hard to realize that the largest set of consecutive elements in $W_{r,s,m}$ is $\{0,1,\ldots, 2^{m+1}-2\}$ and thus, by the remarks about the dual distance after equation \eqref{eq:param_cyclic_dual}, we have that $d(C(W_{r,s,m})^\perp)\geq 2^{m+1}$. 

We define the code
\begin{align*}
	C=[C(W_{r,s,m_1}),C(W_{r,s,m_2})]\begin{bmatrix}
	1 & 1 \\ 0 & 1
	\end{bmatrix},
\end{align*}
where we let $m_1\geq m_2$. Note that $d(C^{\perp})\geq \min\{2\cdot 2^{m_2+1},2^{m_1+1}\}$. From \eqref{eq:inclusion} we conclude that
\begin{align}\label{eq:cyclic_square_s_weight}
	C^{*2}=[C(W_{r,s,m_1}+W_{r,s,m_1}),C(W_{r,s,m_1}+W_{r,s,m_2})]\begin{bmatrix}
	1 & 1 \\ 0 & 1
	\end{bmatrix}
\end{align}
since $W_{r,s,m_1}\cup W_{r,s,m_2}=W_{r,s,m_1}$. It is noted in \cite{C17} that $W_{r,s,m_i}+W_{r,s,m_j}=W_{r,s,m_i+m_j}$ does not hold in general, but the inclusion $W_{r,s,m_i}+W_{r,s,m_j}\subseteq W_{r,s,m_i+m_j}$ holds. However, we are able to determine the exact dimension for $C^{*2}$ in \eqref{eq:cyclic_square_s_weight} by computing $W_{r,s,m_1}+W_{r,s,m_i}$ for $i=1,2$. Additionally, when computed these, we can bound the minimum distance directly from \eqref{eq:dist_square_cyclic}. This is what we do in Table \ref{tab2} for the following choices. We present the parameters for $C$ and $C^{*2}$ when setting $q=2$, $s=5$, $m_1=2$, and $m_2=1$. In this case we have $d(C^\perp)\geq 8$ for each $r$. 
\begin{table}[h]
\centering
\begin{tabular}{|c|c|c|c|c|c|}
\hline
$r$ & $n$ & $\dim(C)$ & $d(C)\geq$ & $\dim(C^{*2})$ & $d(C^{*2})\geq $ \\
\hline
$5$ & $62$ & $22$ & $14$ & $57$ & $2$ 
\\
\hline
$6$ & $126$ & $29$ & $30$ & $99$ & $6$ 
\\
\hline
$7$ & $254$ & $37$ & $62$ & $163$ & $14$ 
\\
\hline
$8$ & $510$ & $54$ & $126$ & $348$ & $18$ 
\\
\hline
$9$ & $1022$ & $86$ & $238$ & $650$ & $38$\\
\hline
$10$ & $2046$ & $142$ & $462$ & $1319$ & $66$\\
\hline
$11$ & $4094$ & $233$ & $926$ & $2543$ & $134$\\
\hline
\end{tabular}
\caption{Parameters for $C$ and $C^{*2}$ using the $5$-restricted weight}
\label{tab2}
\end{table}

We make a comparison to the cyclic codes from \cite{C17}. They present codes constructed using the $3$-restricted weight with $m=1$ (Table 1 in \cite{C17}) and using the $5$-restricted weight with $m=2$ (Table 2 in \cite{C17}). Let any one of our new codes from Table \ref{tab2} have parameters $$(n,\dim(C),d(C^{*2}))=(n,k,d^*).$$ 
First we compare to Table 1 from \cite{C17}, where there always exists a code $C'$ with length $n+1$, $\dim(C')<k$, and $d((C')^{*2})>d^*$. Hence, our new codes have larger dimension but lower minimum distance for the square compared to these codes, for comparable lengths. On the other hand, in Table 2 from \cite{C17} there is a code $C'$ with length $n+1$ and $\dim(C')\geq k$ (i.e. the dimensions of the codes from \cite{C17} are larger than those in our table). However, the minimum distances of the squares for the codes in \cite{C17} satisfy
\begin{align*}
	d((C')^{*2})&=\begin{cases}
		d^*+1 \quad  \text{for } r=5,6,8,10,11\\ 
		d^*-5 \quad  \text{for } r=7,9
	\end{cases}.
\end{align*}
Thus, even though the dimension of our codes are lower than the ones from Table 2 in \cite{C17}, for $r=7$ and $r=9$ we obtain that $d^*>d((C')^{*2})$. 

Therefore, our results on matrix-product codes allow us to obtain codes with a different trade-off between $\dim C$ and $d(C^{*2})$ than those from \cite{C17}, where we can obtain a larger distance of the square at the expense of reducing the dimension with respect to one of the tables there, and viceversa with respect to the other.

\section{Other Matrix-Product Codes}\label{sec:other_cons}
	In this section, we consider squares of some other families of matrix-product codes. We start by determining the square of $C$ when $C$ is a matrix-product code where the defining matrix $A$ is Vandermonde.
	\begin{thm}\label{thm:Vandermonde}
		Let $C_0, C_1, \cdots, C_{s-1}$ be linear codes in $\mathbb{F}_q^n$. Furthermore, let 
		\begin{align*}
		V_q(s)=\begin{bmatrix}
		1 & 1 & \cdots& 1\\
		\alpha_1^1 & \alpha_2^1 & \cdots& \alpha_{q-1}^1\\
		\vdots & \vdots & & \vdots\\
		\alpha_1^{s-1} & \alpha_2^{s-1} & \cdots& \alpha_{q-1}^{s-1}
		\end{bmatrix},
		\end{align*}
		where the $\alpha_i$'s are distinct nonzero elements in $\mathbb{F}_q$ and $s\leq q-1$ is some positive integer. Denote by $C=[C_0,C_1,\ldots,C_{s-1}]V_q(s)$. Then 
		\begin{align*}
		C^{*2}=[\sum_{i+j=0} C_i*C_j,\sum_{i+j=1} C_i*C_j,\ldots,\sum_{i+j=\tilde{s}-1} C_i*C_j]V_q(\tilde{s})
		\end{align*}
		where $\tilde{s}=\min\{2s-1,q-1\}$ and the sums $i+j$ are modulo $q-1$.
	\end{thm}
	\begin{proof}
	 Let $G_0,G_1,\ldots,G_{s-1}$ be generator matrices of $C_0,C_1,\ldots, C_{s-1}$ respectively and let $G$ be a generator matrix for $C$. Using the same notation as in the proof of Theorem \ref{thm_u_u+v_con}, $G*G$ contains all rows of the form
	\begin{align*}
		(\alpha_1^{i+j}G_i*G_j,\ldots,\alpha_q^{i+j}G_i*G_j)
	\end{align*}
	for $i,j=0,1,\ldots,s-1$. Note that if $i+j\geq  q-1$ then $\alpha^{i+j}=\alpha^{i+j-q+1}$ and hence we can consider $i+j$ modulo $q-1$. Thus if $l\leq q-1$ and $i+j\equiv l \pmod{q-1}$ we could write the coefficients in front of $G_i*G_j$ as $\alpha_k^{l}$ for $k=1,2,\ldots,n$. Removing linearly dependent rows this results in a generator matrix for a matrix-product code of the form
	\begin{align}\label{eq:square_vandermonde}
		C^{*2}=[\sum_{i+j=0} C_i*C_j,\sum_{i+j=1} C_i*C_j,\ldots,\sum_{i+j=\tilde{s}-1} C_i*C_j]V_q(\tilde{s}),
	\end{align}
where again $i+j$ is considered modulo $q-1$.
	\end{proof}
	As we will show below, the fact that we obtain codes of the form $\sum_{i+j=l} C_i*C_j$ is especially helpful for determining the parameters of $C^{*2}$ in some cases. We remark that the same phenomenon occurs in the case of the $(u,u+v)$ construction but only if the codes $C_i$ are nested.
	
	Note also that $C(V_q(s)_i)$ (the linear code spanned by the first $i$ rows of $V_q(s)$) is a Reed Solomon code\footnote{A Reed-Solomon code is an MDS code meaning that it achieves the highest possible minimum distance for a given length and dimension. Thus the $D_i$'s are maximal and hence we obtain the best possible bound for the minimum distance we can hope for using the matrix-product construction.} of length $q-1$ and dimension $i$ and hence we have that $d(C(V_q(s)_i))=q-i$, for $i=1,2,\ldots,s$. Applying Proposition \ref{prop:min_dist_bound} with $D_i=d(C(V_q(s)_{i+1}))=q-i-1$ and $d_i=d(C_i)$ for $i=0,1,\ldots,s-1$ we obtain that $C$ is a
	\begin{align*}
		\left[(q-1)n,k_0+k_1+\cdots +k_{s-1},\geq \min_{i\in \{0,1,\cdots, s-1\}}\{(q-i-1)d_i\}\right]_q
	\end{align*}
	linear code, and $C^{*2}$ has minimum distance greater than or equal to
	\begin{align}\label{eq:min_dist_square_vandermonde}
		\min_{l\in \{0,1,\cdots, \tilde{s}-1\}}\Big\{(q-l-1)d\big(\sum_{i+j=l}C_i*C_j\big)\Big\}.
	\end{align}
	Even though the expression in \eqref{eq:square_vandermonde} may at first sight seem hard to work with, this is not the case if the $C_i$'s come from some specific families of codes. For example, Proposition \ref{prop:sum_prod_cyclic} tells us that $\sum_{i+j=l}C_i*C_j$ will again be a cyclic code if the $C_i$'s are cyclic and we will be able to determine its generating set from the generating sets of the $C_i$'s. 
	
	Additionally, one could consider the case where the $C_i$'s are Reed-Solomon codes or more generally algebraic geometric codes. Let $D=P_1+P_2+\cdots+P_n$ be a formal sum of rational places in a function field over $\mathbb{F}_q$ and let $G_i=r_{i,1}Q_1+r_{i,2}Q_2+\cdots+r_{i,m}Q_m$ where all the $Q_i$'s and $P_j$'s are different. An algebraic geometric code $C_i=C_{\mathcal{L}}(D,G_i)$ is the evaluation of the elements in the Riemann-Roch space $\mathcal{L}(G_i)$ in the places from $D$. It is then known that $C_i*C_j\subseteq C_{\mathcal{L}}(D,G_i+G_j)$ and $C_i+C_j\subseteq C_{\mathcal{L}}(D,H)$, where $H=\sum_{k=1}^m \max\{r_{i,k},r_{j,k}\}Q_k$. Hence, we can find a lower bound for $d(C^{*2})$ from \eqref{eq:min_dist_square_vandermonde} using the fact that from the above observations we can find algebraic geometric codes containing $\sum_{i+j=l}C_i*C_j$ where we can control the minimum distance. Furthermore, if $\deg G_i\geq 2g+1$ and $\deg G_j\geq 2g$, where $g$ is the genus of the function field, we obtain that $C_\mathcal{L}(D,G_i)*C_\mathcal{L}(D,G_j)=C_\mathcal{L}(D,G_i+G_j)$, see for instance \cite{CMP17}. Similarly, if we only consider one-point codes, meaning that $G_i=r_iQ$, we obtain that $C_i\subseteq C_j$ if $r_j\geq r_i$ and hence $C_i+C_j=C_j$. We exemplify some specific constructions with algebraic geometric codes, more specific one-point Hermitian codes, in the following example.

\begin{ex}
	We will not go into details about the Hermitian function field and codes, but we do mention that the Hermitian function field is defined over $\mathbb{F}_{q^2}$, it has genus $g=\frac{q(q-1)}{2}$, and it has $q^3+1$ rational places, where one of these places is the place at infinity. Denote the place at infinity by $Q$ and the remaining rational places by $P_i$, for $i=1,2,\ldots,q^3$, and let $D=\sum_{i=1}^{q^3}P_i$. Then a Hermitian code is given by the algebraic geometric code $C_r=C_{\mathcal{L}}(D,rQ)$. This is a $[q^3,r-g+1, q^3-r]_{q^2}$ code as long as $2g\leq r \leq q^3-q^2-1$, see for instance \cite{YK92}. Denote by
	\begin{align*}
		C(r,s)=[C_{r+s-1},C_{r+s-2},\ldots,C_r]V_{q^2}(s),
	\end{align*}
	where $2\leq s\leq\frac{q^2}{2}$ and $2g+1\leq r$. Furthermore, assume that $r+s \leq \frac{q^3-q^2+1}{2}$. With such a construction we have that
	\begin{align}\label{eq:Hermit_square}
		(C(r,s))^{*2}=[C_{2r+2s-2},C_{2r+2s-3},\ldots,C_{2r}]V_{q^2}(2s-1)=C(2r,2s-1)
	\end{align}
	from the observations about algebraic geometric codes above the example. Note that $2r+2s-2\leq q^3-q^2-1$ implying that all the Hermitian codes in \eqref{eq:Hermit_square} satisfy that their $r$ is lower than $q^3-q^2-1$. Hence,
	\begin{align*}
		d((C(r,s))^{*2})		&= \min_{i=0,1\ldots,2s-2} \{(q^2-i-1)(q^3-2r-2s+2+i)\}\\
									&=(q^2-2s+1)(q^3-2r),
	\end{align*}
	where the last equality follows from the following observations:
	\begin{align*}
		&(q^2-i-1)(q^3-2r-2s+2+i)\\
		&=(q^2-2s+1+(2s-2-i))(q^3-2r-(2s-2-i))\\
		&=(q^2-2s+1)(q^3-2r)+(2s-2-i)(q^3-2r-q^2+2s-2-2s+2+i)\\
		&=(q^2-2s+1)(q^3-2r)+(2s-2-i)(q^3-q^2-2r+i).
	\end{align*}
	From the restrictions on $r$ the last factor is positive. Hence, the minimum is attained when $i=2s-2$. Similar arguments show that $d(C(r,s))= (q^2-s)(q^3-r)$. Summing up, we have that $C(r,s)$ has parameters
	\begin{align*}
		\left[q^5-q^3,s\left(r-g+\frac{s+1}{2}\right), (q^2-s)(q^3-r)\right]_{q^2}
	\end{align*}
	and its square is the code $C(2r,2s-1)$ with parameters
	\begin{align*}
		\left[q^5-q^3, (2s-1)\left(2r-g+s \right), (q^2-2s+1)(q^3-2r)\right]_{q^2}.
	\end{align*}
	The parameters of this construction for $q=4$ can be found in Table \ref{tab:Hermit_gen1}.
\begin{table}[ht]
\centering
\begin{tabular}{|c|c|c|c|c|c|c|c|}
\hline
$\text{Field size}$ & $r$& $s$ & $n$ & $k$ & $d $ & $k^* $ & $ d^* $ \\
\hline
$16$ & $13$ & $2$ & $960$ & $17$ & $714$ & $66$ & $494$\\
\hline
$16$ & $16$ & $2$ & $960$ & $23$ & $672$ & $84$ & $416$\\
\hline
$16$ & $19$ & $2$ & $960$ & $29$ & $630$ & $102$ & $338$\\
\hline
$16$ & $22$ & $2$ & $960$ & $35$ & $588$ & $120$ & $260$\\
\hline
$16$ & $13$ & $4$ & $960$ & $38$ & $612$ & $168$ & $342$\\
\hline
$16$ & $16$ & $4$ & $960$ & $50$ & $576$ & $210$ & $288$\\
\hline
$16$ & $19$ & $4$ & $960$ & $62$ & $540$ & $252$ & $234$\\
\hline
$16$ & $13$ & $6$ & $960$ & $63$ & $510$ & $286$ & $190$\\
\hline
$16$ & $16$ & $6$ & $960$ & $81$ & $480$ & $352$ & $160$\\
\hline
$16$ & $13$ & $7$ & $960$ & $77$ & $459$ & $351$ & $114$\\
\hline
$16$ & $16$ & $7$ & $960$ & $98$ & $432$ & $429$ & $96$\\
\hline
$16$ & $13$ & $8$ & $960$ & $92$ & $408$ & $420$ & $38$\\
\hline
$16$ & $16$ & $8$ & $960$ & $116$ & $384$ & $510$ & $32$\\
\hline
\end{tabular}
\caption{Parameters for $C(s,r)$ and $(C(s,r))^{*2}$ with the construction from Theorem \ref{thm:Vandermonde} and Hermitian codes. We use the notation $k=\dim(C(s,r))$, $k^*=\dim((C(s,r))^{*2})$, $d=d(C(s,r))$, and $d^*=d((C(s,r))^{*2})$.}
\label{tab:Hermit_gen1}
\end{table}
~ 
\end{ex}
	Finally, we remark that the strategy used in the proof of Theorems \ref{thm_u_u+v_con} and \ref{thm:Vandermonde} is more or less identical. Thus, such strategy may also be used for other matrix-product codes. For example, we can consider the matrix-product codes with defining matrix
\begin{align*}
	MS_p=\left[\binom{p-i}{j-1} \bmod p\right]_{i,j=1,2,\ldots,p},
\end{align*}
where $p$ is a prime number. We remark that $\binom{n}{k}=0$ if $k>n$ and that these matrix-product codes were also considered in \cite{BN01}. We show that, as for Theorems \ref{thm_u_u+v_con} and \ref{thm:Vandermonde}, we can express the square of such matrix-product codes. 
\begin{thm}
		Let $p$ be a prime number and let $C_1\supseteq C_2\supseteq \cdots\supseteq C_{p}$ be linear codes in $\mathbb{F}_q^n$, where $\mathbb{F}_q$ has characteristic $p$. Denote $C=[C_1,C_2,\ldots,C_{p}]MS_p$ and 
	\begin{align*}
			MS_p^*=\left[\binom{p-i}{j-1}\binom{p-1}{j-1}\bmod p\right]_{i,j=1,2,\ldots,p} .
	\end{align*}
	Then
	\begin{align*}
		C^{*2}=[\sum_{i+j=2} C_i*C_j,\sum_{i+j=3} C_i*C_j,\ldots,\sum_{i+j=p+1} C_i*C_j]MS_p^*.
	\end{align*}
\end{thm}
\begin{proof}
	Let $G_i$ be a generator matrix for $C_i$ and let $G$ be a generator matrix for $C$. As in the previous proofs we construct $G*G$, the matrix whose rows span $C^{*2}$. The rows in $G*G$ are given by
	\begin{align}\label{eq:coeff_GGp}
		\left[\binom{p-m}{j-1}\binom{p-n}{j-1}G_m*G_n\right]_{j=1,2,\ldots,p}
	\end{align}
	for $1\leq m\leq n\leq p$. We remark that the codes are nested and hence we might assume that rows in $G_m*G_n$ are also rows in $G_i*G_n$ when $m\geq i$.
	
	We will show that by doing row operations on $G*G$ we can obtain the rows
	\begin{align}\label{eq:coeff_MSP}
		\left[\binom{p-1}{j-1}\binom{p-m-n+1}{j-1}G_m*G_n\right]_{j=1,2,\ldots,p}.
	\end{align}
	Notice that the coefficients only depends on $m+n$, so proving that we can obtain \eqref{eq:coeff_MSP} by row operations proves the theorem. 
	
	In order to obtain this, we note that we can replace the rows corresponding to $G_m*G_n$ by any linear combination of the form
	\begin{align}\label{eq:linear_comb}
		\sum_{i=1}^m a_i\left[\binom{p-i}{j-1}\binom{p-n}{j-1}G_m*G_n\right]_{j=1,2,\ldots,p}.
	\end{align}
	 
Here we have used that, as remarked in the beginning of the proof, rows in $G_m*G_n$ are included in $G_i*G_n$ when $i\leq m$ and therefore the rows in the linear combination given by \eqref{eq:linear_comb} are included in those from \eqref{eq:coeff_GGp}. 
	
	Thus, in order to end the proof, we need to show that there exist coefficients $a_i\in \mathbb{F}_q$ such that 
	\begin{align}\label{eq:linear_comb_MSp}
		\sum_{i=1}^m a_i \binom{p-i}{j-1}\binom{p-n}{j-1}=\binom{p-1}{j-1}\binom{p-m-n+1}{j-1}, \text{ for } j=1,2,\ldots,p.
	\end{align}
	We observe that if $\binom{p-n}{j-1}=0$ then $\binom{p-n-m+1}{j-1}=0$. Hence, we only need to prove that \eqref{eq:linear_comb_MSp} holds for $j-1\leq p-n$. 
	
	Now note that
	\begin{align*}
		\binom{p-i}{j-1}&=\frac{(p-j)(p-j-1)\cdots(p-i-j+2)}{(p-1)(p-2)\cdots(p-i+1)}\binom{p-1}{j-1}, \quad \text{and}\\
		\binom{p-m-n+1}{j-1}&=\frac{(p-n-j+1)(p-n-j)\cdots(p-n-j-m+3)}{(p-n)(p-n-1)\cdots(p-n-m+2)}\binom{p-n}{j-1}		
	\end{align*}
	Plugging this into \eqref{eq:linear_comb_MSp}  and dividing by $\binom{p-1}{j-1}\binom{p-n}{j-1}$ we see that
	\begin{align*}
		\sum_{i=1}^m a_i\prod_{k=0}^{i-2}\frac{p-j-k}{p-1-k}=\prod_{k=0}^{m-2}\frac{p-n-j+1-k}{p-n-k}.
	\end{align*}
	The right hand side can be considered as a degree $m-1$ polynomial in the variable $j$. The left hand side is a sum of $a_i$ times a polynomial of degree $i-1$ in the variable $j$. Hence, there is a value $a_m$ such that the coefficients of $j^{m-1}$ on both sides coincide. If we set the value of $a_m$, we can then determine the value of $a_{m-1}$ such that the coefficients of $j^{m-2}$ coincide, and so on. In this way we can inductively choose the $a_i$'s such that \eqref{eq:linear_comb_MSp} holds. 
\end{proof}
	We remark that the condition that the codes are nested is essential in the proof since otherwise the linear combination in \eqref{eq:linear_comb} is not valid. 
	
	In summary, we have shown in this section that the squares of some of the most well-known matrix-product codes are again matrix-product codes of a form simple enough that we can relate their minimum distance to that of the squares and products of the constituent codes (in some cases we need those constituent codes to be nested). Furthermore, the strategies used in the proofs from the different constructions are similar, which also suggests that the same techniques may be used for showing similar results in the case of other families of matrix-product codes.
	
\newpage
\appendix
\section{Products and Sums of Cyclic Codes}\label{sec:app}
It is possible to describe a cyclic code as a subfield subcode of an evaluation code. For a set $M\subseteq \{1,\ldots, n-1\}$ and the extension field $\mathfrak{F}=\mathbb{F}_{q}(\mathcal{\beta})$ denote by
\begin{align*}
	\mathcal{P}(M)=\left\{\sum_{i\in M} a_iX^i\mid a_i\in \mathfrak{F}\right\}.
\end{align*}
Furthermore, let $\mathrm{ev}_{\beta}(f)=(f(1),f(\beta),\ldots,f(\beta^{n-1}))$ and define
\begin{align*}
	\mathcal{B}(M)=\{\mathrm{ev}_{\beta}(f)\mid f\in \mathcal{P}(M)\}\subseteq \mathfrak{F}^n.
\end{align*}
This is a linear code over $\mathfrak{F}$. We can obtain a linear code over $\mathbb{F}_q$ by taking the subfield subcode $\mathcal{B}(M)\cap \mathbb{F}_q^n$. Letting $M=-I=\{-i \bmod n\mid i \in I\}$, we obtain the cyclic code with generating set $I$, i.e. 
\begin{align*}
	C(I)=\mathcal{B}(-I)\cap \mathbb{F}_q^n,
\end{align*}
see for instance Lemma 2.22 in \cite{C17}. We generalize Theorem 3.3 in \cite{C17}, stating that $C(I)^{*2}=C(I+I)$, to $C(I_1)*C(I_2)=C(I_1+I_2)$ below. The proofs are almost identical to the ones in \cite{C17}.
\begin{lem}\label{lem:extension_product}
	Let $I_1$ and $I_2$ be cyclotomic cosets. Then
	\begin{align*}
		\mathcal{B}(-I_1)*\mathcal{B}(-I_2)=\mathcal{B}(-(I_1+I_2)),
	\end{align*}
	where $I_1+I_2=\{a+b \bmod n\mid a\in I_1,b\in I_2\}$.
\end{lem}
\begin{proof}
	We start by proving the inclusion $\mathcal{B}(-I_1)*\mathcal{B}(-I_2)\subseteq \mathcal{B}(-(I_1+I_2))$. Let $\vek{v}\in \mathcal{B}(-I_1)*\mathcal{B}(-I_2)$. Then for $f_1\in \mathcal{P}(-I_1)$ and $f_2\in \mathcal{P}(-I_2)$, we have
	\begin{align*}
		\vek{v}=\mathrm{ev}_{\beta}(f_1)*\mathrm{ev}_{\beta}(f_2)=\mathrm{ev}_{\beta}(f_1f_2).
	\end{align*}		
	Since $f_1f_2\in\mathcal{P}(-(I_1+I_2))$ the inclusion follows. 
	
	For the other inclusion, let $\vek{w}\in \mathcal{B}(-(I_1+I_2))$. Then for $f=\sum_{i\in -(I_1+I_2)}a_iX^i\in \mathcal{P}(-(I_1+I_2))$, where $a_i\in \mathfrak{F}$, we have
	\begin{align*}
		\vek{w}=\mathrm{ev}_{\beta}(f)=\sum_{i\in -(I_1+I_2)}a_i \mathrm{ev}_{\beta}(X^i).
	\end{align*}
	Since $i\in -(I_1+I_2)$, there exists $j\in -I_1$ and $k\in -I_2$, such that $X^i=X^jX^k$, and therefore $ \mathrm{ev}_{\beta}(X^i)=\mathrm{ev}_{\beta}(X^j)*\mathrm{ev}_{\beta}(X^k)$. Hence, $\vek{w}$ is an $\mathfrak{F}$-linear combination of elements in $\mathcal{B}(-I_1)*\mathcal{B}(-I_2)$ showing that $\vek{w}\in \mathcal{B}(-I_1)*\mathcal{B}(-I_2)$. 
\end{proof}
With this lemma, we are able to proof Proposition \ref{prop:sum_prod_cyclic}.
	\begin{nonumprop}
		Let $I_1$ and $I_2$ be unions of $q$-cyclotomic cosets, and let $C(I_1)$ and $C(I_2)$ be the corresponding cyclic codes. Then 
		\begin{align*}
			C(I_1)+C(I_2)&=C(I_1\cup I_2),\\
			C(I_1)*C(I_2)&= C(I_1+I_2).
		\end{align*}
		where $I_1+I_2=\{a+b \bmod n\mid a\in I_1,b\in I_2\}$.
	\end{nonumprop}
	\begin{proof}
	 	Let $\tilde{C}$ be the smallest linear code containing $C(I_1)$ and $C(I_2)$. Every codeword $\vek{c}\in \tilde{C}$ must be of the form $a_1\vek{c}_2+a_2\vek{c}_2$, where $a_1,a_2\in \mathbb{F}_q$ and $\vek{c}_i\in C(I_i)$. Now, let $T$ be the function that do a cyclic shift. Then
		\begin{align*}
			T(\vek{c})=T(a_1 \vek{c}_1 + a_2 \vek{c}_2) =a_1T(\vek{c}_1) + a_2T(\vek{c}_2),
		\end{align*}
		and since $T(\vek{c}_i)\in C(I_i)$, we also have $\vek{c}\in \tilde{C}$. Therefore, $\tilde{C}$ is cyclic. For $C(I_i)$ to be included in the cyclic code $\tilde{C}$ the generating polynomial $g$ for $\tilde{C}$ must divide $g_i$, the generating polynomial for $C(I_i)$. The smallest code, which means the polynomial with highest degree, satisfying this is $g=\gcd(g_1,g_2)$. This implies that $\tilde{C}=C(I_1\cup I_2)$.

To show the second equality, we start by noticing that $\mathcal{B}(-I_i)=C(I_i)_{\mathfrak{F}}$, i.e. extending the code $C(I_i)$ to scalars over $\mathfrak{F}$ gives back $\mathcal{B}(-I_i)$. This is shown in \cite{C17}. Additionally, we use Lemma 2.23(iii) in \cite{R15} implying that $C(I_1)_{\mathfrak{F}}*C(I_2)_{\mathfrak{F}}=(C(I_1)*C(I_2))_{\mathfrak{F}}$. Putting these two statements together we obtain
\begin{align*}
	C(I_1)*C(I_2)	&=(C(I_1)*C(I_2))_{\mathfrak{F}} \cap \mathbb{F}_q^n=(C(I_1)_\mathfrak{F}*C(I_2)_\mathfrak{F})\cap \mathbb{F}_q^n\\
						&=(\mathcal{B}(-I_1)*\mathcal{B}(-I_2))\cap \mathbb{F}_q^n=\mathcal{B}(-(I_1+I_2))\cap \mathbb{F}_q^n\\
						&=C(I_1+I_2),
\end{align*}
where we have used Lemma \ref{lem:extension_product} in the second to last step.
	\end{proof}
\newpage
\printbibliography
\end{document}